\documentclass[12pt]{amsart}

\usepackage[centertags]{amsmath}
\usepackage{amsthm,amsfonts}
\usepackage{tikz, float}

\usepackage{bm}

\usepackage{amssymb}
\usepackage{epsfig}
\usepackage{amssymb, latexsym}%, dsfont, mathrsfs}
\usepackage{indentfirst}

\usepackage{color}
\usepackage{ulem}
\usepackage{verbatim}
\begin{document}
%%%%%%%%%%%%%%%%%%%% Text italic %%%%%%%%%%%%%%%%%%%%%%%%%%%%
\theoremstyle{plain}
\newtheorem{theorem}{Theorem}[section]
\newtheorem{lemma}{Lemma}[section]
\newtheorem{proposition}{Proposition}[section]
\newtheorem{corollary}{Corollary}[section]
\newtheorem{example}{Example}[section]
%%%%%%%%%%%%%%%%%%%% Text roman %%%%%%%%%%%%%%%%%%%%%%%%%%%%%
\theoremstyle{definition}
\newtheorem{notations}[theorem]{Notations}
\newtheorem{notation}[theorem]{Notation}
\newtheorem{remark}[theorem]{Remark}
\newtheorem{question}[theorem]{Question}
\newtheorem{definition}[theorem]{Definition}
\newtheorem{condition}[theorem]{Condition}

\newcommand{\Hom}{{\rm Hom}}
\newcommand{\Ima}{{\rm Im}}
\newcommand{\Ker}{{\rm Ker \ }}
\newcommand{\rank}{{\rm rank}}
\newcommand{\Norm}{{\rm Norm}}
\newcommand{\rowspace}{{\rm rowspace}}
\newcommand{\Span}{{\rm Span}}
\newcommand{\Tr}{{\rm Tr}}
\newcommand{\Char}{{\rm char \ }}

\newcommand{\soplus}[1]{\stackrel{#1}{\oplus}}
\newcommand{\dlog}{{\rm dlog}\,}    % For dlog

\newenvironment{pf}{\noindent\textbf{Proof.}\quad}{\hfill{$\Box$}}

\newcommand{\sA}{{\mathcal A}}
\newcommand{\sC}{{\mathcal C}}
\newcommand{\sD}{{\mathcal D}}
\newcommand{\sG}{{\mathcal G}}
\newcommand{\sH}{{\mathcal H}}
\newcommand{\sL}{{\mathcal L}}
\newcommand{\sP}{{\mathcal P}}

\newcommand{\F}{{\mathbb F}}
\newcommand{\Z}{{\mathbb Z}}
\newcommand{\R}{{\mathbb R}}
\newcommand{\N}{{\mathbb N}}
\newcommand{\A}{{\mathbb A}}
\newcommand{\K}{{\mathbb K}}

%%%%%%%%%%%%%%%%%%%%%%%%%%%%%%%%%%%%%%%
\newcommand{\be}{\begin{eqnarray}}
\newcommand{\ee}{\end{eqnarray}}
\newcommand{\nn}{{\nonumber}}
\newcommand{\dd}{\displaystyle}
\newcommand{\ra}{\rightarrow}

\newcommand{\mi}{{\rm max-isometry}}

\newcommand{\SqBinom}[2]{\genfrac{[}{]}{0pt}{}{#1}{#2}}
\newcommand{\Binom}[2]{\genfrac{(}{)}{0pt}{}{#1}{#2}}

\newcommand{\Keywords}[1]{\par\noindent
{\small{\it Keywords\/}: #1}}

\title[LCD codes tridiagonal Toeplitz matrices]{LCD  Codes from tridiagonal Toeplitz matrices}

\author[Minjia Shi   \& Ferruh \"{O}zbudak \& Li Xu \& Patrick Sol\'e]{Minjia Shi \& Ferruh \"{O}zbudak \& Li Xu\&  Patrick Sol\'e}
%\thanks{MSC codes: 94B14, 11T71, 94B27}

\thanks{Minjia Shi is with Key Laboratory of Intelligent Computing Signal Processing, Ministry of Education,
School of Mathematical Sciences, Anhui University, Hefei, Anhui, 230601, China, e-mail: smjwcl.good@163.com}

\thanks{Ferruh \"{O}zbudak is with Department of Mathematics and Institute of Applied Mathematics, Middle East Technical
        University,   Ankara, Turkey;
        e-mail: ozbudak@metu.edu.tr}
        
\thanks{Li Xu is with School of Mathematical Sciences, Anhui University, Hefei, Anhui, 230601, China,  e-mail: xuli1451@163.com}

\thanks{Patrick Sol\'e is with I2M, Aix Marseille Univ., Centrale Marseille, CNRS, Marseille, France , e-mail: sole@enst.fr}

\abstract
Double Toeplitz (DT) codes are codes with a generator matrix of the form $(I,T)$ with $T$ a Toeplitz matrix, that is to say constant on the diagonals parallel to the main.
When $T$ is tridiagonal and symmetric we determine its spectrum explicitly by using Dickson polynomials, and deduce from there conditions for the code to be LCD.
Using a special concatenation process, we construct optimal or quasi-optimal examples of binary and ternary LCD codes from DT codes over extension fields.
\vspace{0.7cm}


{\em AMS(2020) Math Sc. Cl.} 94B05, 15B05, 12E10

\endabstract

\maketitle
%\tableofcontents

\section{Introduction}

Linear Complementary Dual (LCD) codes are linear codes which intersect their dual trivially. They were introduced by Massey in 1992 to solve a problem in Information Theory \cite{Massey}.
They were proved to be asymptotically good by Sendrier \cite{S1}, who used them in relation with equivalence testing of linear codes \cite{S2}. They enjoyed a renewal of interest in 2016, with
an application to side-channel attacks on embarked cryptosystems \cite{CG}. Recently LCD double circulant codes or double negacirculant codes were constructed over various alphabets \cite{HSS,SHSS1,SHSS2,1,2}. A far reaching generalization
of both double circulant and double negacirculant codes is that of double Toeplitz codes \cite{SXS}. In the present paper, we introduce a class of double Toeplitz codes which can be effectively tested for being LCD.

A code is double Toeplitz (DT) if its generator matrix is of the form $(I,T)$ with $I$ an identity matrix, and $T$ a Toeplitz matrix of the same order. Recall that a matrix is Toeplitz if
it has constant entries on all diagonals parallel to the main diagonal. Thus circulant matrices and negacirculant matrices are Toeplitz.

It is easy to check that such a code is LCD iff $-1$ is not an eigenvalue of $TT^t.$ To make that condition easy to check we will
make two hypotheses on $T:$
\begin{itemize}
 \item $T=T^t$ implying $TT^t=T^2;$
 \item $T$ is {\em tridiagonal}, in  the sense that $T_{ij}=0$ if $|i-j|>1.$
\end{itemize}

In the next section, we show that the characteristic polynomial of a tridiagonal symmetric Toeplitz matrix satisfies a three-term recurrence that can be identified, up to an easy
change of variable to that of the Dickson polynomials \cite{DicksonBook}. The roots of these polynomials can be determined explicitly \cite{BZ}. Hence we obtain an exact and explicit characterization on whether a given DT code $(I,T)$ is LCD or not, when $T$ is tridiagonal and symmetric (see Theorems \ref{theorem.LCD.char.even.extension} and \ref{theorem.LCD.char.odd.extension} below).
It seems very difficult to obtain such a characterization for arbitrary Toeplitz $T$.  Moreover this is the first paper in the literature using factorization of Dickson polynomials for the characterization of some LCD codes as far as we know.

Under some mild arithmetic conditions
we can show that this spectrum does not intersect  the base field, and in particular does not contain $-1.$ Some sufficient conditions for the DT code to be LCD follow.
Since the DT codes so constructed have minimum distance at most three, a rather sophisticated concatenation process, namely isometry (see Definition \ref{isometry} below) can be used to construct an LCD code over a small field. Note that because of the fundamental result
that any linear code
over $\F_q$ with $q>3$ is equivalent to an LCD code \cite{CMTQP}, the theory of LCD codes is focusing on the cases of $\F_2$ and $\F_3.$ Using the said concatenation process optimal
or quasi-optimal LCD codes over these two fields are explicitly constructed.

The material is organized as follows. The next section studies the spectrum of Toeplitz matrices. Section 3 describes a concatenation process that allow for LCD codes over small fields.
Numerical examples are given there. The last section concludes the paper.
\section{Toeplitz matrices}
\subsection{A Spectral lemma}
Throughout this paper, let $p$ be a prime, $q=p^s$ for a positive integer $s$.
Let $\F_q$ denote the finite field of $q$ elements. Let $\overline{\F}_q$ denote an agebraic closure of $\F_q$.

\begin{lemma} \label{lemma1}
For $n \ge 1$ let $A$ be an $n \times n$ matrix over $\F_q$. We have the following cases:
\begin{itemize}
\item  \underline{$\Char \F_q$ is even:}
$-1$ is an eigenvalue of $A^2$ if and only if $-1$ is an eigenvalue of $A$.
\item  \underline{$\Char \F_q$ is odd:}
$-1$ is an eigenvalue of $A^2$ if and only if $-\mu$ or $\mu$ is an eigenvalue of $A$,
where $\mu \in \F_{q^2}$ with $\mu^2=-1$.
\end{itemize}
\end{lemma}
\begin{proof}
If $\Char \F_q$ is even, then
\be
(A+I_n)^2=A^2+I_n^2=A^2+I_n,
\nn\ee
which completes the proof in this case.

If $\Char \F_q$ is odd, then
\be
(A+\mu I_n)(A-\mu I_n)= A^2 -\mu^2 I_n = A^2 + I_n,
\nn\ee
which completes the proof.
\end{proof}
\subsection{ Characteristic polynomial}
For $a \in \F_q$ and $n \ge 3$, let $T_n(a)$ be the $n \times n$ {\it tridiagonal Toeplitz} matrix depending on $a$ defined as
\be
T_n(a)=\left[
\begin{array}{cccccc}
a & 1 & 0 & \cdots & & \\
1 & a & 1 & \cdots & & \\
\vdots & & & & & \\
0 & & & \cdots & 1 & a
\end{array}
\right].
\nn\ee
Namely, for example, we have
\be
T_3(a)=\left[
\begin{array}{ccc}
a & 1 & 0  \\
1 & a & 1  \\
0 & 1 & a
\end{array}
\right]
\;\mbox{and} \;
T_4(a)=\left[
\begin{array}{cccc}
a & 1 & 0 & 0 \\
1 & a & 1 & 0 \\
0 & 1 & a & 1 \\
0 & 0 & 1 & a
\end{array}
\right].
\nn\ee
We also define the cases for $n=1,2$ as
\be
T_1(a)=\left[
\begin{array}{c}
a
\end{array}
\right]
\; \mbox{and} \;
T_2(a)=\left[
\begin{array}{cc}
a & 1 \\
1 & a
\end{array}
\right].
\nn\ee

For $n \ge 1$ let
\be
\phi_n(\lambda)=\det\left( T_n(a)- \lambda I_n\right).
\nn\ee

\begin{lemma} \label{lemma2}
Under notation as above, we have
\be
\phi_n(\lambda)= (a-\lambda) \phi_{n-1}(\lambda) - \phi_{n-2}(\lambda)
\nn\ee
for $n \ge 2$ with $\phi_1(\lambda)=a-\lambda$ and $\phi_0(\lambda)=1$.
\end{lemma}
\begin{proof}
For $n=2$ we have
\be
T_2(a)-\lambda I_2=\left[
\begin{array}{cc}
a-\lambda & 1 \\
1 & a - \lambda
\end{array}
\right]
\nn\ee
and hence
\be
\phi_2(\lambda)= \det\left(T_2(a)-\lambda I_2\right)=(a-\lambda)^2 -1 = (a-\lambda)\phi_1(\lambda) - \phi_0(\lambda).
\nn\ee
This completes the proof for $n=2$.

For $n=3$ we have
\be
T_3(a)-\lambda I_3= \left[
\begin{array}{ccc}
a-\lambda & 1 & 0 \\
1 & a-\lambda & 1 \\
0 & 1 & a-\lambda
\end{array}
\right].
\nn\ee
Considering the expansion of $\det\left(T_3(a)-\lambda I_3\right)$ using the last row we obtain
\be
\phi_3(\lambda)= (a-\lambda)
\left|
\begin{array}{cc}
a-\lambda & 1 \\
1 & a- \lambda\end{array}
\right|
-
\left|
\begin{array}{cc}
a-\lambda & 0 \\
1 & 1 \end{array}
\right|
=(a-\lambda) \phi_2(\lambda) - \phi_1(\lambda).
\nn\ee
 This completes the proof for $n=3$.

For $n \ge 3$ we will show that the lemma holds for $n+1$. Note that this will complete the proof. Assume that $n \ge 3$. For the
$(n+1) \times (n+1)$ tridiagonal Toeplitz  matrix $T_{n+1}(a)$ depending on $a$
we have
\be
T_{n+1}(a)=\left[
\begin{array}{ccc}
T_{n-1}(a) & C_{n-1} & 0_{(n-1)\times 1} \\
R_{n-1} & a & 1 \\
0_{1 \times (n-1)} & 1 & a
\end{array}
\right].
\nn\ee
Here $T_{n-1}(a)$ is the $(n-1) \times (n-1)$ triagonal Toeplitz matrix depending on $a$, $0_{1 \times (n-1)}$ is the $1 \times (n-1)$ matrix whose all entries are $0$
and $0_{(n-1) \times 1}$ is the $(n-1) \times 1$ matrix whose all entries are $0$. Moreover, $R_{n-1}$ is an $1 \times (n-1)$ row matrix and $C_{n-1}$ is an $(n-1) \times 1$ column matrix.

Considering the expansion of $\det\left( T_{n+1}(a)-\lambda I_{n+1}\right)$ using the last row we obtain
\be
\begin{array}{rcl}
\dd \phi_{n+1}(\lambda) & = &  \dd
\left|
\begin{array}{ccc}
T_{n-1}(a) - \lambda I_{n-1} & C_{n-1} & 0_{(n-1)\times 1} \\
R_{n-1} & a-\lambda & 1 \\
0_{1\times (n-1)} & 1 & a-\lambda
\end{array}
\right|
\\ \\
& = &  \dd (a-\lambda)
\left|
\begin{array}{cc}
T_{n-1}(a) - \lambda I_{n-1} & C_{n-1} \\
R_{n-1} & a-\lambda
\end{array}
\right|
-
\left|
\begin{array}{cc}
T_{n-1}(a) -\lambda I_{n-1} & 0_{(n-1)\times 1} \\
R_{n-1} & 1
\end{array}
\right| \\ \\
 & = & \dd (a-\lambda)
\left|
\begin{array}{cc}
T_{n-1}(a) - \lambda I_{n-1} & C_{n-1} \\
R_{n-1} & a-\lambda
\end{array}
\right|
- \phi_{n-1}(a).
\end{array}
\nn\ee
In the last equality we use the expansion of $\left|
\begin{array}{cc}
T_{n-1}(a) -\lambda I_{n-1} & 0_{(n-1)\times 1} \\
R_{n-1} & 1
\end{array}
\right|$
using the last column. Note that
\be
\left[
\begin{array}{cc}
T_{n-1}(a) - \lambda I_{n-1} & C_{n-1} \\
R_{n-1} & a-\lambda
\end{array}
\right]
=T_n(a)-\lambda I_n.
\nn\ee
Hence we conclude that
\be
\phi_{n+1}(\lambda)=(a-\lambda)\phi_n(\lambda) - \phi_{n-1}(\lambda)
\nn\ee
for $n \ge 3$, which completes the proof.
\end{proof}

For  $n \ge0$, $a \in \F_q$ and $x \in \overline{\F}_q$, let $\psi_n$ be the function on $\overline{\F}_q$ defined as
\be
\begin{array}{rcc}
\psi_n:\overline{\F}_q & \ra & \overline{\F}_q \\
x & \mapsto & \phi_n(a-x).
\end{array}
\nn\ee

\begin{lemma} \label{lemma3}
Under notation as above we have
\be
\psi_n(x)= x\psi_{n-1}(x) - \psi_{n-2}(x)
\nn\ee
for $n \ge 2$ with $\psi_1(x)=x$ and $\psi_0(x)=1$.
\end{lemma}
\begin{proof}
For $n=1$ and $n=0$ we have
\be
\psi_1(x)=\phi_1(a-x)=a-(a-x)=x, \; \mbox{and} \; \psi_0(x)=\phi_0(a-x)=1.
\nn\ee
For $n \ge 2$ we have
\be
\begin{array}{rcl}
\psi_n(x) & = & \phi_n(a-x) \\
& = &\left(a-(a-x)\right) \phi_{n-1}(a-x) - \phi_{n-2}(a-x), \;\; \mbox{using Lemma \ref{lemma2}}, \\
& = & x \phi_{n-1}(a-x)-\phi_{n-2}(a-x) \\
& = & x \psi_{n-1}(x) - \psi_{n-2}(x), \;\; \mbox{by definitions of $\psi_{n-1}(x)$ and $\psi_{n-2}(x)$.}
\end{array}
\nn\ee
This completes the proof.
\end{proof}

Recall that (see, for example, \cite{DicksonBook}) for $\alpha \in \F_q$ and $n \ge 0$, the Dickson polynomial of the second kind
\be
E_n(x,\alpha) \in \F_q[x]
\nn\ee
is defined recursively
\be
E_n(x,\alpha)=xE_{n-1}(x,\alpha) - \alpha E_{n-2}(x,\alpha)
\nn\ee
with the initial conditions $E_1(x,\alpha)=x$ and $E_0(x)=1$.

We put $\alpha=1$ and denote $E_n(x)=E_n(x,1)$ throughout the paper. We obtain that $E_n(x) \in \F_q[x]$ is the polynomial defined resursively
\be \label{def.En}
E_n(x)=xE_{n-1}(x)-E_{n-2}(x)
\ee
with the initial conditions $E_1(x)=x$ and $E_0(x)=1$.

Combining Lemma \ref{lemma2}, Lemma \ref{lemma3} and (\ref{def.En}) we prove the following proposition immediately.

\begin{proposition} \label{proposition1}
Under notation as above we have
\be
\det\left( T_n(a)-xI_n\right)=\phi_n(x)=E_n(a-x),
\nn\ee
for all $n \ge 1$ and $x \in \overline{\F}_q$.
\end{proposition}

We recall the following result due to Bhargava and Zieve \cite[Theorem 4]{BZ}.

\begin{theorem} \label{theorem1.BZ}
Let $n \ge 1$ be an integer. Assume that $\gcd(n+1,q)=1$. We have the following cases
\begin{itemize}
\item \underline{$\Char \F_q$ is odd:}
\be
E_n(x)=\prod_{i=1}^n \left( x - \left( \theta^i + \theta^{-i}\right)\right),
\nn\ee
where $\theta$ is a primitive $2(n+1)$-th root of $1$.

\item  \underline{ $\Char \F_q$ is even:}  then $n$ is even and
\be
E_n(x)=\prod_{i=1}^{n/2} \left( x - \left( \theta^i + \theta^{-i}\right)\right)^2,
\nn\ee
where $\theta$ is a primitive $(n+1)$-th root of $1$.
\end{itemize}
\end{theorem}

\begin{definition} \label{defition.Cna}
For $a \in \F_q$ and an integer $n \ge 1$, let $C_n(a)$ be the $\F_q$-linear code of length $2n$ and dimension $n$ whose generator polynomial is the $n\times 2n$ matrix given by
\be
\left[ I_n  \; | \; T_n(a) \right].
\nn\ee
\end{definition}

\subsection{ LCD double Toeplitz codes}

Now we present our first characterization result on whether $C_n(a)$ is LCD or not.

\begin{theorem} \label{theorem.LCD.char.even}
For $a \in \F_q$ and an integer $n \ge 1$, consider the $[2n,n]_q$ code $C_n(a)$ given in Definition \ref{defition.Cna}. Assume that $\gcd(n+1,q)=1$ and $\Char \F_q$ is even. Then $n$ is even and $C_n(a)$ is LCD if and only if
\be
a \not \in \left\{-1 + \theta^i + \theta^{-i}: 1 \le i \le \frac{n}{2}\right\},
\nn\ee
where $\theta$ is a primitive $(n+1)$-th root of $1$.
\end{theorem}
\begin{proof}
It is well known that $C_n(a)$ is LCD (see, for example, \cite{Massey}) if and only if $GG^T$ is invertible, where $G=\left[ I_n  \; | \; T_n(a) \right]$. Note that
\be
GG^T=\left[ I_n  \; | \; T_n(a) \right] \left[\begin{array}{c} I_n \\
T_n(a)
\end{array}
\right]= I_n + T_n(a)^2,
\nn\ee
where we use the fact that $T_n(a)$ is symmetric. Hence $C_n(a)$ is LCD iff $-1$ is not an eigenvalue of $T_n^2(a)$. Using Lemma \ref{lemma1} we conclude that $C_n(a)$ is LCD iff
$-1$ is not an eigenvalue of $T_n(a)$. It follows from the definition of $\phi_n(x)$ that $-1$ is an eigenvalue of $T_n(a)$ iff $\phi_n(-1)=0$. Using Proposition \ref{proposition1} we have that $\phi_n(-1)=E_n(a+1)$. Finally using Theorem \ref{theorem1.BZ} we complete the proof.
\end{proof}

In odd characteristic we have the following result.

\begin{theorem} \label{theorem.LCD.char.odd}
For $a \in \F_q$ and an integer $n \ge 1$, consider the $[2n,n]_q$ code $C_n(a)$ given in Definition \ref{defition.Cna}. Assume that $\gcd(n+1,q)=1$ and $\Char \F_q$ is odd. Then
$C_n(a)$ is LCD if and only if
\be
a \not \in \left\{-\mu + \theta^i + \theta^{-i}: 1 \le i \le n\right\} \cup
\left\{\mu + \theta^i + \theta^{-i}: 1 \le i \le n\right\},
\nn\ee
where $\mu^2=-1$ and $\theta$ is a primitive $2(n+1)$-th root of $1$.
\end{theorem}
\begin{proof}
The proof is similar to that of Theorem \ref{theorem.LCD.char.even}. We only indicate the different steps in this proof. Note that $-1$ is an eigenvalue of $T^2_n(a)$ iff
$-\mu$ or $\mu$ is an eigenvalue of $T_n(a)$ by Lemma \ref{lemma1}. Hence $C_n(a)$ is not LCD iff $E_n(a+\mu)=0$ or $E_n(a-\mu)=0$. We complete the proof using the similar steps as in the proof of Theorem \ref{theorem.LCD.char.even}.
\end{proof}
The following corollaries are immediate.

\begin{corollary} \label{corollary.LCD.char.even}
Assume that $\gcd(n+1,q)=1$ and $\Char \F_q$ is even. If $q > \frac{n}{2}$, then there exists $a \in \F_q$ such that $C_n(a)$ is LCD.
\end{corollary}
\begin{proof}
Let $S=\{-1 + \theta^i + \theta^{-i}: 1 \le i \le n/2\}$, where $\theta$ is a primitive $(n+1)$-th roof of $1$. Note that $|S| \le n/2$ and hence $|S \cap \F_q| \le n/2$. As $q> n/2$, there exists $a \in \F_q \setminus S$. Using such $a$ and Theorem \ref{theorem.LCD.char.even} we complete the proof.
\end{proof}

\begin{corollary} \label{corollary.LCD.char.odd}
Assume that $\gcd(n+1,q)=1$ and $\Char \F_q$ is odd. If $q >2n$, then there exists $a \in \F_q$ such that $C_n(a)$ is LCD.
\end{corollary}
\begin{proof}
Let $S=\{-\mu + \theta^i + \theta^{-i}: 1 \le i \le n\} \cup \{\mu + \theta^i + \theta^{-i}: 1 \le i \le n\}$, where $\mu^2=-1$ and $\theta$ is a primitive $2(n+1)$-th root of $1$. Note that $|S| \le 2n$ and hence $|S \cap \F_q| \le 2n$. As $q> 2n$, there exists $a \in \F_q \setminus S$. Using such $a$ and Theorem \ref{theorem.LCD.char.odd} we complete the proof.
\end{proof}

Corollary \ref{corollary.LCD.char.even} and Corollary \ref{corollary.LCD.char.odd} give
that there exist $a \in \F_q$ so that $C_n(a)$ is LCD for infinitely many $C_n(a)$, namely if $q$ is large compared to $n$.
The corresponding conditions are not necessary for the existence.
Under arithmetic conditions, the following two corollaries give simple existence results for even and odd characteristics.

\begin{corollary} \label{cor.arith.existence.q.even}
Assume  that $q$ is even and $\gcd(n+1,q(q^2-1))=1$.
We have that $C_n(a)$ is LCD for all $a \in \F_q.$
\end{corollary}

\begin{proof}
 Note that $\gcd(n+1,q)=1$ and $\gcd(n+1,q^2-1)=1$. Let $\theta$ be a primitive $(n+1)$-th root of $1$.
Let $1 \le i \le n/2$ be an arbitrary integer.
 Put $t=\theta^i.$ We have $ 1+t+1/t \in \F_q$ iff $( 1+t+1/t)^q= 1+t+1/t$ iff $u^q=u$ with $u=t+1/t.$
Thus $u$ is in $\F_q$ and $t$ is in $\F_{q^2}.$ Since $t^{n+1}=1,$ with $t\neq 1,$ this is impossible for  $\gcd(n+1,q^2-1)=1.$
The result follows by Theorem \ref{theorem.LCD.char.even}.
\end{proof}

The analog of Corollary \ref{cor.arith.existence.q.even} for the odd characteristic is slightly more complicated.

\begin{corollary} \label{cor.arith.existence.q.odd}
Assume that $q$ is odd. Moreover, we assume the following:
\begin{itemize}
\item If $q \equiv 1 \mod 4$, then $\gcd(n+1, q(q^2-1)/2)=1$. \\
\item If $q \equiv 3 \mod 4$, then $\gcd(n+1, q)=1$ and $\gcd(n+1,(q^4-1)/2)$ divides $(q-1)/2$.
\end{itemize}
Let $\mu \in \F_{q^2}$ such that $\mu^2=-1$. We have the following:
\begin{itemize}
\item If $q \equiv 1 \mod 4$, then $C_n(a)$ is LCD for all $a\in \F_q \setminus \{\mu+2,-\mu+2,\mu-2,-\mu-2\}$. \\
\item If $q \equiv 3 \mod 4$, then $C_n(a)$ is LCD for all $a\in \F_q$.
\end{itemize}
\end{corollary}

\begin{proof}
Note that $\gcd(n+1,q)=1$. Let $\theta$ be a primitive $2(n+1)$-th root of $1$.
Let $1 \le i \le n$ be an arbitrary integer.
Put $t=\theta^i.$ Also put $  w=\pm \mu + t + 1/t$. Assume that $ w \in \F_q$. Next we consider the cases separately:
\begin{itemize}
\item \underline{ $q \equiv 1 \mod 4$:} Here $\mu \in \F_q$ and hence $t+1/t \in \F_q$. This implies that $t \in \F_{q^2}$ and $t^{q^2-1}=1$. As $\gcd(q^2-1,2(n+1))=2$ by assumption, we obtain that $t \in \{-1,1\}$. Hence we have $\{\mu +t + 1/t, -\mu +t + 1/t\}  = \{\mu+2,\mu-2,-\mu+2,-\mu-2\}$. We complete the proof using Theorem \ref{theorem.LCD.char.odd}. \\
\item \underline{ $q \equiv 3 \mod 4$:} Here $\mu \in \F_{q^2} \setminus \F_q$ and hence $t+1/t \in \F_{q^2}$. This implies that $t \in \F_{q^4}$ and $t^{q^4-1}=1$. As $\gcd(q^4-1,2(n+1))$ divides $q-1$ by assumption, we obtain that $t \in \F_q$. Hence we have $\pm\mu +t + 1/t\in \F_{q^2} \setminus \F_q$. We complete the proof using Theorem \ref{theorem.LCD.char.odd}.
\end{itemize}
This completes the proof.\end{proof}

We also give simple examples that are not covered by these corollaries.
\begin{example} \label{example.LCD.odd.1}
Let $q=3$ and $n=3$ so that the condition $q> 2n$ of Corollary \ref{corollary.LCD.char.odd} does not hold.
Also the conditions of Corollary \ref{cor.arith.existence.q.odd} do not hold.
Let $\theta$ be a primitive $2(n+1)=8$-th root of $1$. Let $\mu=\theta^2$ so that $\mu^2=-1$.
Let $S=\{-\mu + \theta^i+\theta^{-i}: 1 \le i \le n\} \cup \{\mu + \theta^i+\theta^{-i}: 1 \le i \le n\}$. Using Magma \cite{Magma} we obtain that
$S=\{0, \theta^2,\theta^6\}$. By Theorem \ref{theorem.LCD.char.odd} we conclude that $C_n(a)$ is LCD for $a \in \F_q$ iff $a\in\{1, 2\}$.
\end{example}

\begin{example} \label{example.LCD.odd.2}
Let $q=3$ and $n=4$ so that the condition $q> 2n$ of Corollary \ref{corollary.LCD.char.odd} does not hold.
Also the conditions of Corollary \ref{cor.arith.existence.q.odd} do not hold.
Let $\theta$ be a primitive $2(n+1)=10$-th root of $1$. Let $\mu=\theta^5$ so that $\mu^2=-1$.
Let $S=\{-\mu + \theta^i+\theta^{-i}: 1 \le i \le n\} \cup \{\mu + \theta^i+\theta^{-i}: 1 \le i \le n\}$.
Let $w$ be a primitive element of $\F_{3^4}$ such that $w^4+2w^3+2=0$.
Using Magma \cite{Magma} we obtain that
$S=\{w^{10},w^{20},w^{30},w^{50},w^{60},w^{70}\}$. Note that $\F_q=\{0,w^{40},w^{80}\}$. By Theorem \ref{theorem.LCD.char.odd} we conclude that $C_n(a)$ is LCD for any $a \in \F_q$.
\end{example}

\subsection{Extension of the results for $T(a,b)$}

For $a,b \in \F_q$ and $n \ge 3$, let $\hat{T}_n(a,b)$ be the $n \times n$ {\it triagonal Topelitz} matrix depending on $a$ and $b$ defined as
\be
\hat{T}_n(a,b)=\left[
\begin{array}{cccccc}
a & b & 0 & \cdots & & \\
b & a & b & \cdots & & \\
\vdots & & & & & \\
0 & & & \cdots & b & a
\end{array}
\right].
\nn\ee
Namely, for example, we have
\be
\hat{T}_3(a,b)=\left[
\begin{array}{ccc}
a & b & 0  \\
b & a & b  \\
0 & b & a
\end{array}
\right]
\;\mbox{and} \;
\hat{T}_4(a,b)=\left[
\begin{array}{cccc}
a & b & 0 & 0 \\
b & a & b & 0 \\
0 & b & a & b \\
0 & 0 & b & a
\end{array}
\right].
\nn\ee
We also define the cases for $n=2$ as
\be
\hat{T}_2(a,b)=\left[
\begin{array}{cc}
a & b \\
b & a
\end{array}
\right].
\nn\ee

It is easy to observe that  if $b \neq 0$, then
\be \label{hat.T.relation}
\frac{1}{b} \hat{T}_n(a,b)= T_n(a/b),
\ee
for $n \ge 2$ and $a \in \F_q$.
This observation leads to the following.

\begin{lemma}  \label{lemma4}
For $n \ge 2$, $a,b \in \F_q$ and $b \neq 0$ we have that
$\lambda$ is an eigenvalue of $\hat{T}_n(a,b)$ if and only if
$\lambda/b$ is an eigenvalue of $T_n(a/b)$.
\end{lemma}
\begin{proof}
For $\lambda \in \overline{\F}_q$ we have that
\be
\begin{array}{rcl}
\det\left( \hat{T}_n(a,b) - \lambda I_n\right)=0 & \iff & \det\left( \frac{1}{b}\left(\hat{T}_n(a,b) - \lambda I_n\right) \right)=0 \\
& \iff & \det\left( T_n(a/b) - \lambda/b I_n\right)=0,
\end{array}
\nn\ee
where we use (\ref{hat.T.relation}). This completes the proof.
\end{proof}

Combining Theorem \ref{theorem.LCD.char.even} and Lemma \ref{lemma4} we immediately obtain the following generalization.

\begin{definition} \label{defition.Cna2}
For $a,b \in \F_q$ and an integer $n \ge 2$, let $\hat{C}_n(a,b)$ be the $\F_q$-linear code of length $2n$ and dimension $n$ whose generator polynomial is the $n\times 2n$ matrix given by
\be
\left[ I_n  \; | \; \hat{T}_n(a,b) \right].
\nn\ee
\end{definition}

\begin{theorem} \label{theorem.LCD.char.even2}
For $a, b \in \F_q$ with $b \neq 0$  and an integer $n \ge 2$, consider the $[2n,n]_q$ code $\hat{C}_n(a,b)$ given in Definition \ref{defition.Cna2}. Assume that $\gcd(n+1,q)=1$ and $\Char \F_q$ is even. Then $n$ is even and $\hat{C}_n(a,b)$ is LCD if and only if
\be
a/b \not \in \left\{-1/b + \theta^i + \theta^{-i}: 1 \le i \le \frac{n}{2}\right\},
\nn\ee
where $\theta$ is a primitive $(n+1)$-th root of $1$.
\end{theorem}
\begin{proof}
Note that $-1$ is an eigenvalue of $\hat{T}_n(a,b)$ iff $-1/b$ is an eigenvalue of $T_n(a/b)$ by Lemma \ref{lemma4}. This holds iff $E_n((a+1)/b)=0$ by Proposition \ref{proposition1}. We complete the proof using Theorem \ref{theorem1.BZ} as in the proof of Theorem \ref{theorem.LCD.char.even}.
\end{proof}

\begin{theorem} \label{theorem.LCD.char.odd2}
For $a, b \in \F_q$ with $b \neq 0$  and an integer $n \ge 2$, consider the $[2n,n]_q$ code $\hat{C}_n(a,b)$ given in Definition \ref{defition.Cna2}.
Assume that $\gcd(n+1,q)=1$ and $\Char \F_q$ is odd. Then
$\hat{C}_n(a,b)$ is LCD if and only if
\be
a/b \not \in \left\{-\mu/b + \theta^i + \theta^{-i}: 1 \le i \le n\right\} \cup
\left\{\mu/b + \theta^i + \theta^{-i}: 1 \le i \le n\right\},
\nn\ee
where $\mu^2=-1$ and $\theta$ is a primitive $2(n+1)$-th root of $1$.
\end{theorem}

\begin{proof}
The proof is similar to those of Theorem \ref{theorem.LCD.char.even2} and Theorem \ref{theorem.LCD.char.odd}.
Note that $-\mu$ or $\mu$ is an eigenvalue of $\hat{T}_n(a,b)$ iff $-\mu/b$ or $\mu/b$ is an eigenvalue of $T_n(a/b)$ by Lemma \ref{lemma4}.
 This holds iff $E_n((a+\mu)/b)=0$ or $E_n((a-\mu)/b)=0$ by Proposition \ref{proposition1}. We complete the proof using Theorem \ref{theorem1.BZ} as in the proof of Theorem \ref{theorem.LCD.char.odd}.
\end{proof}

\subsection{Extension of the results for $\gcd(n+1,q) \neq 1$}

 First we introduce further notation: For positive integers $a,b$ and a nonnegative integer $u$, $b^u \mid\mid a$ denotes that $b^u \mid $a and $b^{u+1} \nmid a$. In the following theorem we recall the result in \cite[Theorem 4]{BZ} for  the arbitrary case including $\gcd(n+1,q) \neq 1$.

\begin{theorem} \label{theorem2.BZ}
Let $n \ge 1$ be an integer.
For the characteristic $p$ of $\F_q$, let $r$ be the nonnegative integer such that $p^r \mid\mid (n+1)$. Let $m$ be the nonnegative integer such that $n+1=p^r(m+1)$. We have the following:
\begin{itemize}
\item \underline{$p$ is odd:}
\be
E_n(x)=E_m(x)^{p^r}(x-2)^{(p^r-1)/2}(x+2)^{(p^r+1)/2} \;\; \mbox{and} \;\;
E_m(x)=\prod_{i=1}^m \left( x - \left( \theta^i + \theta^{-i}\right)\right),
\nn\ee
where $\theta$ is a primitive $2(m+1)$-th root of $1$.

\item \underline{  $p=2$:} then $m$ is even and
\be
E_n(x)=E_m(x)^{2^r}x^{2^r-1} \;\; \mbox{and} \;\;
E_m(x)=\prod_{i=1}^{m/2} \left( x - \left( \theta^i + \theta^{-i}\right)\right)^2,
\nn\ee
where $\theta$ is a primitive $(m+1)$-th root of $1$.
\end{itemize}
\end{theorem}

We extend Theorem \ref{theorem.LCD.char.even2} (and hence Theorem \ref{theorem.LCD.char.even}) for $\gcd(n+1,q) \neq 1$.

\begin{theorem} \label{theorem.LCD.char.even.extension}
For $a, b \in \F_q$ with $b \neq 0$  and an integer $n \ge 2$, consider the $[2n,n]_q$ code $\hat{C}_n(a,b)$ given in Definition \ref{defition.Cna2}. Assume that $\Char \F_q$ is even.
Let $r$ be the nonnegative integer such that $2^r \mid\mid (n+1)$. Let $m$ be the nonnegative integer such that $n+1=2^r(m+1)$. Assume that $r \ge 1$ (see Theorem \ref{theorem.LCD.char.even2} for the remaining case that $r=0$). We have that $m$ is even. If $m>0$, then
$\hat{C}_n(a,b)$ is LCD if and only if
\be
a/b \not \in \{-1/b\} \cup \left\{-1/b + \theta^i + \theta^{-i}: 1 \le i \le \frac{m}{2}\right\},
\nn\ee
where $\theta$ is a primitive $(m+1)$-th root of $1$. If $m=0$, then
$\hat{C}_n(a,b)$ is LCD if and only if
$a\neq 1$.
\end{theorem}
\begin{proof}
The proof is similar to that of Theorem \ref{theorem.LCD.char.even2}. The main difference is that we use
Theorem \ref{theorem2.BZ} instead of Theorem \ref{theorem1.BZ}.
\end{proof}

Now we extend Theorem \ref{theorem.LCD.char.odd2} (and hence Theorem \ref{theorem.LCD.char.odd}) for $\gcd(n+1,q) \neq 1$.

\begin{theorem} \label{theorem.LCD.char.odd.extension}
For $a, b \in \F_q$ with $b \neq 0$  and an integer $n \ge 2$, consider the $[2n,n]_q$ code $\hat{C}_n(a,b)$ given in Definition \ref{defition.Cna2}. Assume that $\Char \F_q$ is odd, which is $p$.
Let $r$ be the nonnegative integer such that $p^r \mid\mid (n+1)$. Let $m$ be the nonnegative integer such that $n+1=p^r(m+1)$. Assume that $r \ge 1$ (see Theorem \ref{theorem.LCD.char.odd2} for the remaining case that $r=0$).  If $m>0$, then
$\hat{C}_n(a,b)$ is LCD if and only if
\be
a/b & \not \in
\begin{array}{l}
\left\{-\mu/b+2, -\mu/b-2, \mu/b+2, \mu/b-2 \right\}
\cup
\left\{-\mu/b + \theta^i + \theta^{-i}: 1 \le i \le m\right\}
\\ \cup
\left\{\mu/b + \theta^i + \theta^{-i}: 1 \le i \le m\right\},
\end{array}
\nn\ee
where $\mu^2=-1$ and $\theta$ is a primitive $2(m+1)$-th root of $1$. If $m=0$, then
$\hat{C}_n(a,b)$ is LCD if and only if $a/b \not \in \{-\mu/b +2, -\mu/b -2, \mu/b +2, \mu/b-2\}$.
\end{theorem}
\begin{proof}
The proof is similar to that of Theorem \ref{theorem.LCD.char.odd2}. The main difference is that we use
Theorem \ref{theorem2.BZ} instead of Theorem \ref{theorem1.BZ}.
\end{proof}

Next we extend Corollaries \ref{corollary.LCD.char.even}, \ref{corollary.LCD.char.odd}, \ref{cor.arith.existence.q.even} and \ref{cor.arith.existence.q.odd}.

\begin{corollary} \label{corollary.LCD.char.even.extension}
Let $b \in \F_q$ with $b \neq 0$  and $n \ge 2$ an integer.
Assume that $\Char \F_q$ is even.
Let $r$ be the nonnegative integer such that $2^r \mid\mid (n+1)$. Let $m$ be the nonnegative integer such that $n+1=2^r(m+1)$. Assume that $r \ge 1$ (see Corollary \ref{corollary.LCD.char.even} for the remaining case that $r=0$).  If $q>m/2+1$, then there exists $a \in \F_q$ such that
$\hat{C}_n(a,b)$ is LCD.
\end{corollary}
\begin{proof}
The proof is similar to that of Corollary \ref{corollary.LCD.char.even}. We have $q>m/2+1$ instead of $q > m/2$ in the hypothesis as $x=0$ is a root of $E_n(x)$ for $r \ge 1$.
\end{proof}

\begin{corollary} \label{corollary.LCD.char.odd.extension}
Let $b \in \F_q$ with $b \neq 0$  and $n \ge 2$ an integer.
Assume that $\Char \F_q$ is odd.
Let $r$ be the nonnegative integer such that $2^r \mid\mid (n+1)$. Let $m$ be the nonnegative integer such that $n+1=2^r(m+1)$. Assume that $r \ge 1$ (see Corollary \ref{corollary.LCD.char.odd} for the remaining case that $r=0$).  If $q>2m+4$, then there exists $a \in \F_q$ such that
$\hat{C}_n(a,b)$ is LCD.
\end{corollary}
\begin{proof}
The proof is similar to that of Corollary \ref{corollary.LCD.char.odd}. We have $q>2m+4$ instead of $q > 2m$ in the hypothesis as $x=2$ and $x=-2$ are roots of $E_n(x)$ for $r \ge 1$.
\end{proof}

\begin{corollary} \label{cor.arith.existence.q.even.extension}
Let $b \in \F_q$ with $b \neq 0$  and $n \ge 2$ an integer.
Assume that $\Char \F_q$ is even.
Let $r$ be the nonnegative integer such that $2^r \mid\mid (n+1)$. Let $m$ be the nonnegative integer such that $n+1=2^r(m+1)$. Assume that $r \ge 1$ (see Corollary \ref{cor.arith.existence.q.even} for the remaining case that $r=0$).
Furthermore, assume that $\gcd(m+1,q^2-1)=1$. We have that
$\hat{C}_n(a,b)$ is LCD for all $a \in \F_q \setminus \{1\}$.
\end{corollary}
\begin{proof}
The proof is similar to that of Corollary \ref{cor.arith.existence.q.even}. We have $\hat{C}_n(a,b)$ is LCD for all $a \in \F_q \setminus \{1\}$ instead of  for all $a \in \F_q $ in the conclusion as $\hat{C}_n({1,b})$ is not LCD by Theorem \ref{theorem.LCD.char.even.extension} when $r \ge 1$.
\end{proof}
In the following corollary there are no differences in the conclusion for the cases $q \equiv 1 \mod 4$ and $q \equiv 3 \mod 4$, which is not the situation in Corollary \ref{cor.arith.existence.q.odd}.
\begin{corollary} \label{cor.arith.existence.q.odd.extension}
Let $b \in \F_q$ with $b \neq 0$  and $n \ge 2$ an integer.
Assume that $\Char \F_q$ is odd, which is $p$.
Let $r$ be the nonnegative integer such that $p^r \mid\mid (n+1)$. Let $m$ be the nonnegative integer such that $n+1=p^r(m+1)$.  Assume that $r \ge 1$ (see Corollary \ref{cor.arith.existence.q.odd} for the remaining case that $r=0$).
Furthermore, assume that
\begin{itemize}
\item If $q \equiv 1 \mod 4$, then $\gcd(m+1, (q^2-1)/2)=1$. \\
\item If $q \equiv 3 \mod 4$, then $\gcd(m+1,(q^4-1)/2)$ divides $(q-1)/2$.
\end{itemize}
Let $\mu \in \F_{q^2}$ such that $\mu^2=-1$. We have that
$\hat{C}_n(a,b)$ is LCD for all $a \in \F_q \setminus \{\mu+2b,\mu-2b,-\mu+2b,\mu-2b\}$.
\end{corollary}
\begin{proof}
The proof is similar to the proof of Corollary \ref{cor.arith.existence.q.odd}. We have $\hat{C}_n(a,b)$ is not LCD
for $a \in \{\mu+2b,\mu-2b,-\mu+2b,\mu-2b\}$
as $\hat{C}_n(a,b)$ is not LCD if $a/b \in \{\mu/b+2,\mu/b-2,-\mu/b+2,-\mu/b-2\}$ by Theorem \ref{theorem.LCD.char.odd.extension} when $r \ge 1$.
\end{proof}

\section{Concatenation}

In this section we construct LCD codes over $\F_q$ with prescribed large minimum distance using DT that we characterize in Theorems \ref{theorem.LCD.char.even.extension} and \ref{theorem.LCD.char.odd.extension} over an extension field $\F_{q^s}$ and a kind of concatenation. It is not difficult to observe that most of the concatenation maps do not work as they would not respect LCD property over the base and the extension fields. Hence we use an isometry map, which is introduced in \cite{CGOS} as a special concatenation respecting LCD property. The minimum distance of the DT codes in Theorems \ref{theorem.LCD.char.even.extension} and \ref{theorem.LCD.char.odd.extension} have minimum distance at most $3$. However, the minimum distance of the isometry code can be arbitrarily large, provided the length of the isometry code is increased if necessary.

First we recall some results and notations from \cite{CGOS}.

\begin{definition} \label{isometry}
Let $n \ge s \ge 2$ be integers. An $\F_q$-linear map $\pi: \F_{q^s} \ra \F_q^n$ is called an {\it isometry} if there exists a basis $(e_1, \ldots, e_s)$ of $\F_{q^s}$ over $\F_q$ such that
\be
\pi(e_i) \cdot \pi(e'_j)=\delta_{i,j}
\nn\ee
for all $1 \le i,j \le s$. Here $\cdot$ is the Euclidean inner product on $\F_q^n$, $(e'_1, \ldots, e'_s)$ is the dual basis of $(e_1, \ldots, e_s)$, and $\delta_{i,j}$ is the Kronecker delta.

The image $\pi(\F_{q^s})$ is an $[n,s]_q$ code, which we call an isometry code. Let $d_\mi(q;[n,s])$ be the largest nonnegative integer $d$ such that there exists an isometry $\pi:\F_{q^s} \ra \F_q^n$ and $\pi(\F_{q^s})$ has minimum distance $d$.
\end{definition}

Note that $d_\mi(q;[n,s])$ coincides with the largest maximum distance of $[n,s]_q$ codes for many parameters. For example
$d_\mi(2;[4,2])=2$, $d_\mi(2;[5,3])=2$ and $d_\mi(3;[5,2])=3$ (see \cite{CGOS}).

Let $s \ge 2$ be an integer. Let $n \ge s$ be an integer such that $d_\mi(q,[n,s]) \ge 1$. Let $\pi:\F_{q^s}
\ra \F_q^n$ be an isometry such that $\pi(\F_{q^s})$ is an $[n,s,d]_q$ code, where $d=d_\mi(q,[n,s])$.

For $a,b \in \F_{q^s}$ with $b \neq 0$ and an integer $N \ge 2$, let $\hat{C}_N(a,b)$ be the $[2N,N]_{q^s}$ code given in Definition \ref{defition.Cna2}.

Let $\pi^{\otimes 2N}: \F_{q^s}^{2N} \ra \F_q^{2Nn}$ be the $\F_q$-linear map given by
\be
\pi^{\otimes 2N}(c_1, \ldots, c_{2N})=\left(\pi(c_1), \ldots, \pi(c_{2N})\right).
\nn\ee
We use these notations in the following two theorems.

Now we are ready to construct LCD codes of arbitrary minimum distance using tridiagonal Toeplitz matrices over extension fields and isometry. First we present the even characteristic case.

\begin{theorem} \label{theorem.isometry.even}
Under notations as above assume that
$\Char \F_q$ is even.
Let $r$ be the nonnegative integer such that $2^r \mid\mid (N+1)$. Let $m$ be the nonnegative integer such that $N+1=2^r(m+1)$. Let $\theta$ be a primitive $(m+1)$-th root of $1$.  Moreover we assume the following in the corresponding cases:
\begin{itemize}
\item If $r=0$, then $a/b \not \in \{1/b + \theta^i + \theta^{-i}: 1 \le i \le m/2\}$. \\
\item If $r \ge 1$ and $m=0$, then $a \neq 1$. \\
\item If $r \ge 1$ and $m \ge 1$, then $a/b \not \in \{1/b\} \cup \{1/b + \theta^i + \theta^{-i}: 1 \le i \le m/2\}$.
\end{itemize}

Then $\pi^{\otimes 2N}\left( \hat{C}_N(a,b) \right)$ is an LCD code with  parameters $[2nN,sN,D^*]_q$ such that $D^* \ge dD$, where $D$ is the minimum distance of $\hat{C}_N(a,b)$.
\end{theorem}
\begin{proof}
Using Theorem \ref{theorem.LCD.char.even.extension} we obtain that $\hat{C}_N(a,b)$ is an LCD code over the extension field $\F_{q^s}$ with minimum distance $D$. As $\pi(\F_{q^s})$ is an $[n,s,d]_q$ isometry code, we complete the proof by \cite[Theorem 3.1]{CGOS}.
\end{proof}

Next we consider the odd characteristic case.

\begin{theorem} \label{theorem.isometry.odd}
Under notation as above assume that
$\Char \F_q$ is odd, which is $p$.
Let $r$ be the nonnegative integer such that $p^r \mid\mid (N+1)$. Let $m$ be the nonnegative integer such that $N+1=p^r(m+1)$. Let $\theta$ be a primitive $2(m+1)$-th root of $1$. Let $\mu^2 =-1$. Moreover, we assume the following in the corresponding cases:
\begin{itemize}
\item If $r=0$, then
\\$
a/b \not \in
\begin{array}{l}
\left\{-\mu/b + \theta^i + \theta^{-i}: 1 \le i \le m\right\}
\cup
\left\{\mu/b + \theta^i + \theta^{-i}: 1 \le i \le m\right\}.
\end{array}
$
\\

\item If $r \ge 1$ and $m=0$, then
\\$a/b \not \in
\left\{-\mu/b+2, -\mu/b-2, \mu/b+2, \mu/b-2 \right\}
$. \\

\item If $r \ge 1$ and $m \ge 1$, then
\\$a/b \not \in
\begin{array}{l}
\left\{-\mu/b+2, -\mu/b-2, \mu/b+2, \mu/b-2 \right\}
\cup
\left\{-\mu/b + \theta^i + \theta^{-i}: 1 \le i \le m\right\}
\\ \cup
\left\{\mu/b + \theta^i + \theta^{-i}: 1 \le i \le m\right\},
\end{array}
$.
\end{itemize}
Then $\pi^{\otimes 2N}\left( \hat{C}_N(a,b) \right)$ is an LCD code with  parameters $[2nN,sN,D^*]_q$ such that $D^* \ge dD$, where $D$ is the minimum distance of $\hat{C}_N(a,b)$.
\end{theorem}
\begin{proof}
The proof is similar to that of Theorem \ref{theorem.isometry.even}. The difference is that we use Theorem \ref{theorem.LCD.char.odd.extension} instead of Theorem \ref{theorem.LCD.char.even.extension}.
\end{proof}

In the following examples, we illustrate how to construct LCD codes with a prescribed
lower bound on the minimum distance over small fields, in particular $\F_2$ and $\F_3$, using
the methods of this paper. In fact, we obtain good codes having optimal and almost
optimal parameters and the actual minimum distances of our constructed codes are
even better than the prescribed lower bounds in these examples.

\begin{example} \label{example1}
Let $s=2$, $q=2$ and $N=4$. Let $w$ be a primitive element of $\F_{q^s}$ satisfying $w^2+w+1=0$. For $a=w$ and $b=1$, the $[2N,N]_{q^s}$ code  $\hat{C}_N(a,b)$ given in Definition \ref{defition.Cna2} is an LCD code having parameters $[4,2,3]_4$ (see Theorem \ref{theorem.LCD.char.even.extension}).
Put $n=4$.
For $a_1=w$, $a_2=w^2$, $a_3=1$ and $a_4=1$, the $\F_q$-linear map
\be
\begin{array}{rcl}
\pi: \F_{q^s} & \ra & \F_q^n \\
x & \mapsto & \left( \Tr(a_1 x), \Tr(a_2 x), \Tr(a_3 x), \Tr(a_4 x) \right)
\end{array}
\nn\ee
is an isometry map such that the corresponding isometry code $\pi(\F_{q^s})$ is an $[4,2,2]_2$ code.
Here $\Tr$ is the trace map from $\F_{q^s}$ onto $\F_q$. Using Theorem \ref{theorem.isometry.even} we obtain that
$\pi^{\otimes 2N}\left( \hat{C}_N(a,b) \right)$
is an LCD code with  parameters $[16,4,D^*]_2$ with the prescribed lower bound on the minimum distance $D^*$ given by $D^* \ge 6$. In fact using Magma \cite{Magma} it is easy to verify  that $D^* = 7$. This is an optimal LCD code, namely the largest minimum distance $D$ of LCD code with parameters $[16,4,D]_2$ is $7$ (see \cite{AH}).
\end{example}

\begin{example} \label{example2}
Let $s=3$, $q=2$ and $N=4$. Let $w$ be a primitive element of $\F_{q^s}$ satisfying $w^3+w+1=0$. For $a=w$ and $b=w^6$, the $[2N,N]_{q^s}$ code  $\hat{C}_N(a,b)$ given in Definition \ref{defition.Cna2} is an LCD code having parameters $[4,2,3]_4$ (see Theorem \ref{theorem.LCD.char.even.extension}).
Put $n=5$.
For $a_1=w^3$, $a_2=w^5$, $a_3=w^6$, $a_4=1$ and $a_5=1$, the $\F_q$-linear map
\be
\begin{array}{rcl}
\pi: \F_{q^s} & \ra & \F_q^n \\
x & \mapsto & \left( \Tr(a_1 x), \Tr(a_2 x), \Tr(a_3 x), \Tr(a_4 x), \Tr(a_5 x) \right)
\end{array}
\nn\ee
is an isometry map such that the corresponding isometry code $\pi(\F_{q^s})$ is an $[5,3,2]_2$ code.
Here $\Tr$ is the trace map from $\F_{q^s}$ onto $\F_q$. Using Theorem \ref{theorem.isometry.even} we obtain that
$\pi^{\otimes 2N}\left( \hat{C}_N(a,b) \right)$
is an LCD code with  parameters $[20,6,D^*]_2$ with the prescribed lower bound on the minimum distance $D^*$ given by $D^* \ge 6$. In fact using Magma \cite{Magma} it is easy to verify  that in fact $D^* = 7$. This is an almost optimal LCD code, namely  the largest minimum distance $D$ of and LCD code  with parameters $[20,6,D]_2$ is $8$ (see \cite{AH}).
\end{example}

\begin{example} \label{example3}
Let $s=2$, $q=3$ and $N=4$. Let $w$ be a primitive element of $\F_{q^s}$ satisfying $w^2+2w+2=0$. For $a=2$ and $b=w$, the $[2N,N]_{q^s}$ code  $\hat{C}_N(a,b)$ given in Definition \ref{defition.Cna2} is an LCD code having parameters $[4,2,3]_9$ (see Theorem \ref{theorem.LCD.char.odd.extension}).
Put $n=5$.
For $a_1=w$, $a_2=w$, $a_3=w^3$, $a_4=w^3$ and $a_5=2$, the $\F_q$-linear map
\be
\begin{array}{rcl}
\pi: \F_{q^s} & \ra & \F_q^n \\
x & \mapsto & \left( \Tr(a_1 x), \Tr(a_2 x), \Tr(a_3 x), \Tr(a_4 x), \Tr(a_5 x) \right)
\end{array}
\nn\ee
is an isometry map such that the corresponding isometry code $\pi(\F_{q^s})$ is an $[5,2,3]_3$ code.
Here $\Tr$ is the trace map from $\F_{q^s}$ onto $\F_q$. Using Theorem \ref{theorem.isometry.even} we obtain that
$\pi^{\otimes 2N}\left( \hat{C}_N(a,b) \right)$
is an LCD code with  parameters $[20,4,D^*]_3$ with the prescribed lower bound on the minimum distance $D^*$ given by $D^* \ge 9$.
Using Magma \cite{Magma} it is easy to verify  that in fact $D^* = 10$.
The largest minimum distance $D$ of an LCD code  with parameters $[20,4,D]_3$ is $12$
by \cite[Table 7]{AH+}.
\end{example}

\section{Conclusion}
In this paper we have constructed LCD double Toeplitz codes from tridiagonal symmetric Toeplitz matrices. It would be worthwhile to extend these results to symmetric Toeplitz matrices with
more than three nontrivial diagonals. We conjecture that this might require multivariate Dickson polynomials \cite{DicksonBook}. This might help to construct DT codes over small fields without
recourse to the concatenation process of the previous section.

\section*{Acknowledgement}
This research is supported by the National Natural Science Foundation of China
  (Grants no. 12071001 and 61672036), the Excellent Youth Foundation
  of Natural Science Foundation of Anhui Province (1808085J20), the
  Academic Fund for Outstanding Talents in Universities (gxbjZD03).

%\newpage

%\section*{Appendix}

\end{document}